\documentclass[lettersize,journal]{IEEEtran}
\usepackage{amsmath,amsfonts}
\usepackage{amsthm}
\usepackage{array}
\usepackage[caption=false,font=normalsize,labelfont=sf,textfont=sf]{subfig}
\usepackage{textcomp}
\usepackage{stfloats}
\usepackage{url}
\usepackage{verbatim}
\usepackage{graphicx}
\usepackage{cite}
\newtheorem{theorem}{Theorem}[section]

\newtheorem{lemma}[theorem]{Lemma}

\begin{document}
	\title{Performance evaluation of conditional handover in 5G systems under fading scenario}
\author{Souvik Deb, Megh Rathod, Rishi Balamurugan, Shankar K. Ghosh, Rajeev K. Singh and Samriddha Sanyal
\thanks{Megh Rathod, Rishi Balamurugan, Shankar K. Ghosh and Rajeev 
Kumar Singh are with the Department of Computer Science and Engineering, Shiv Nadar Institution of Eminence, Delhi NCR, India. Shankar K. Ghosh will act as corresponding author. Emails: shankar.ghosh@snu.edu.in, shankar.it46@gmail.com}
\thanks{Souvik Deb is with Advanced Computing and Microelectronics Unit, Indian Statistical Institute, Kolkata, India, Email: deb.souvik5@gmail.com}
\thanks{Samriddha Sanyal is with the School of Computing and Data Science, FLAME University, Pune, India, Email: samriddha.sanyal@flame.edu.in}
}

\maketitle
\begin{abstract}
To enhance the handover performance in fifth generation (5G) cellular systems, conditional handover (CHO) has been evolved as a promising solution. Unlike A3 based handover where handover execution is certain after receiving handover command from the serving access network, in CHO, handover execution is \emph{conditional} on the RSRP measurements from both current and target access networks, as well as on mobility parameters such as preparation and execution offsets. Analytic evaluation of conditional handover performance is unprecedented in literature. In this work, handover performance of CHO has been carried out in terms of handover latency, handover packet loss and handover failure probability. A Markov model accounting the effect of different mobility parameters (e.g., execution offset, preparation offset, time-to-preparation and time-to-execution), UE velocity and channel fading characteristics; has been proposed to characterize handover failure. Results obtained from the analytic model has been validated against extensive simulation results. Our study reveal that optimal configuration of $O_{exec}$, $O_{prep}$, $T_{exec}$ and $T_{prep}$ is actually conditional on underlying UE velocity and fading characteristics. This study will be helpful for the mobile operators to choose appropriate thresholds of the mobility parameters under different channel condition and UE velocities.
\end{abstract}


\section{Introduction}
To enhance capacity in fifth generation (5G) cellular system, dense deployment of small cells has been evolved as an well-accepted solution.  While roaming across such dense deployment scenario, user equipment (UEs) have to perform handover from the serving next generation Node B (gNB) to target gNB to sustain connectivity with the system. In the traditional A3 condition based handover \cite{SON_handover_parameters}, the handover execution phase immediately follow the handover preparation phase, which may not be appropriate for millimeter wave (mmWave) communication. This is because mmWave communication is highly susceptible to environmental losses (e.g., path loss due to high frequency, presence of dynamic obstacle etc.), and therefore the channel condition of the target gNB may change right after the handover preparation. As a result, the handover execution may fail leading to handover failure (HOF). 
\begin{figure*}[t]
    \centering
    \includegraphics[scale=0.8]{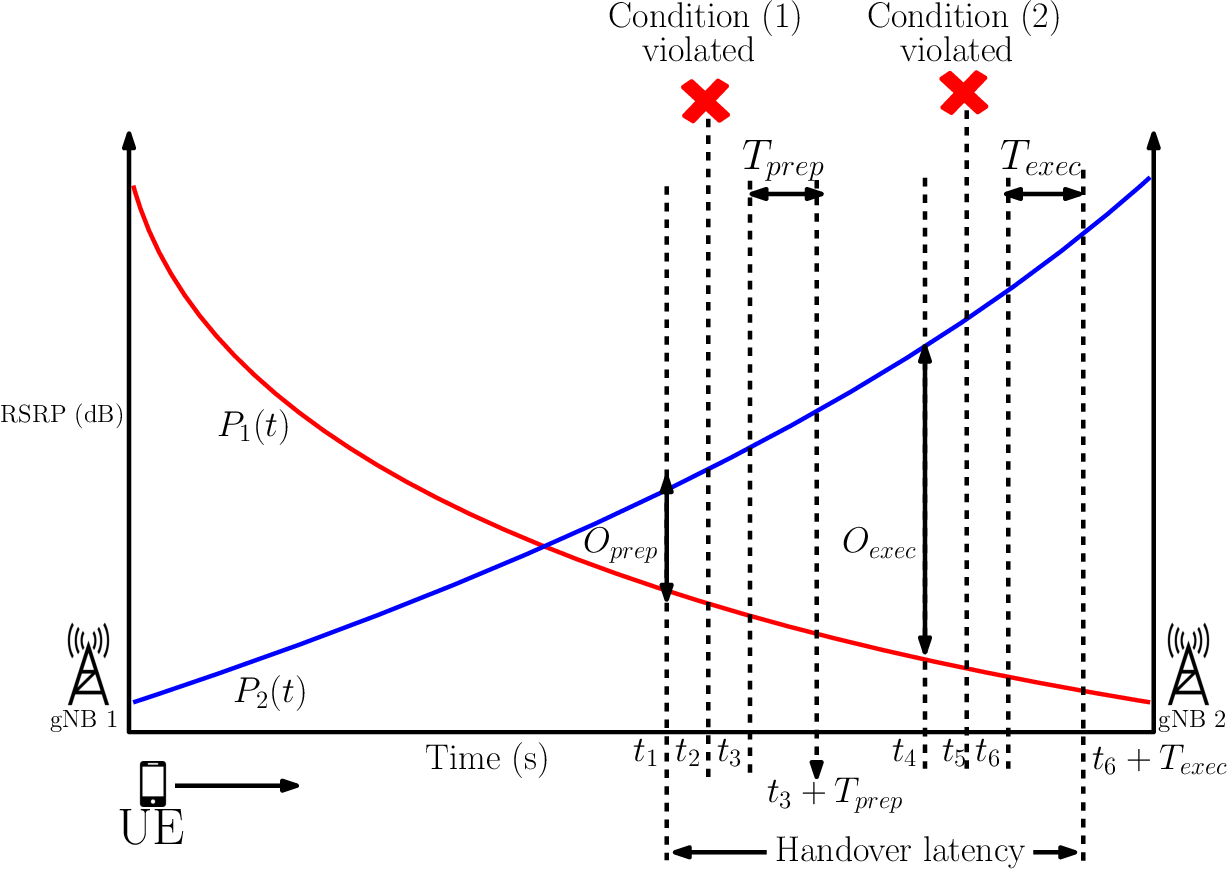}
    \caption{Demonstration of Conditional handover in fading scenario}
    \label{cho_mech}
\end{figure*}

To alleviate the aforesaid effect, conditional handover (CHO) \cite{bsics_of_cho} has been proposed. The idea of CHO is to decouple the handover preparation from handover execution phase. In CHO, the user equipment (UE) may delay the execution phase after handover preparation based on the channel condition of the target gNB, and thereby reducing both radio link failure (RLF) and subsequent HOF. To illustrate, let us consider that an UE is moving from serving gNB (say gNB 1) to target gNB (say gNB 2) with constant velocity through the linear trajectory connecting the centers of gNB 1 and gNB 2 (Fig. \ref{cho_mech}). For handover preparation to be completed, the UE must satisfy the following condition \cite{bsics_of_cho}:
\begin{equation}\label{prep_offset}
   P_{2}(t') > P_{1}(t')+O_{prep},
\end{equation}
\noindent for $t\le t'\le t+T_{prep}$. Here $P_1(t)$ is the RSRP received at the UE at time $t$ from gNB 1, $P_2(t)$ is the RSRP received at the UE at time $t$ from gNB 2, $O_{prep}$ is the  handover preparation offset and $T_{prep}$ is the predefined time-to-prepare.
After completion of the handover preparation event, the next phase is handover execution. For the execution phase to begin, the UE waits for the following condition to be satisfied:

\begin{equation}\label{exec_offset}
     P_{2}(t')> P_{1}(t')+O_{exec},  
\end{equation}
\noindent for $t\le t'\le t+T_{exec}$. Here, $O_{exec}$ is the offset for handover execution phase and $T_{exec}$ is a predefined time-to-execute. 

The UE keeps on measuring $P_1(t)$ and $P_2(t)$ samples periodically in every $20$ ms time duration through synchronization signal block (SSB) bursts transmitted by gNBs \cite{iqbal2023analysis}. The RSRP measurements are stochastic in nature due to fading. For the preparation phase to be satisfied, the condition that $P_{2}(t') > P_{1}(t')+O_{prep}$ has to be satisfied for each and every sampling instances during $T_{prep}$. Similarly, for the execution phase to be satisfied, the condition that $P_{2}(t') > P_{1}(t')+O_{exec}$ has to be true for each and every sampling instance during $T_{exec}$. In Fig \ref{cho_mech}, the condition \eqref{prep_offset} is first satisfied at time $t^{'}=t_1$. But due to fading, condition \eqref{prep_offset} is violated in a subsequent sample taken at $t^{'}=t_2$ and therefore cancelling the ongoing handover preparation event. Condition \eqref{prep_offset} is again true at time $t^{'}=t_3$, and therefore starting the handover preparation phase once again. This time, condition \eqref{prep_offset} is true for all sampling instances between $t_3$ and $t_3 + T_{prep}$, therefore completing the handover preparation phase successfully at $t_3 + T_{prep}$. After handover preparation, the UE is waiting for condition \eqref{exec_offset} to be satisfied to start the handover execution phase. The condition \eqref{exec_offset} is satisfied first time at $t^{'}=t_4$, but subsequently the ongoing handover execution is cancelled at $t^{'}=t_5$ due to fading. Condition \eqref{exec_offset} is again satisfied at $t^{'}=t_6$ (restart of handover execution phase), and holds for all the sampling instances between $t_6$ and $t_6 + T_{exec}$. Therefore, the CHO process is completed successfully at $t_6 + T_{exec}$. It may be noted that due to fading, the handover process is delayed by ($t_3$-$t_1$)+($t_6$+$t_4$) time duration, and therefore increasing the handover latency ($t_6 + T_{exec}$-$t_1$) as shown in Fig. \ref{cho_mech}. Such a delay may cause RLF, and subsequent HOF.

\subsection{Motivation and our contributions}



To illustrate the effect of fading on handover latency, RLF and subsequent HOF, we have conducted a system level simulation considering the 2-gNB model shown in Fig. \ref{cho_mech}. Complete code of the simulator can be found in \cite{github}. Here we have considered standard simulation parameter settings \cite{iqbal2023analysis}: $T_{prep}=100$ ms, $T_{exec}=80$ ms, $O_{exec}=3$ dB, $O_{prep}=10$ dB. Figs \ref{nofad_rlf} and \ref{nofad_hof} show higher RLF and HOF under fading scenario (both Raleigh and Rician) compared to deterministic path-loss (no fading) scenario. This is because of the prolonged handover preparation and execution phase due to fading as shown in Fig. \ref{cho_mech}. The RSRP fluctuates more frequently in case of Raleigh compared to Rician counterpart because the Rician fading has a line of sight component making the RSRP more stable. As a result, RLF and HOF in Raleigh is higher compared to that of Rician fading. Fig. \ref{lv} shows that handover latency ($t_6 + T_{exec}-t_1$) in case of Raleigh fading is much higher compared to that of no fading case. This is because of the time wasted ($t_3-t_1$)+($t_6-t_4$) due to violation of conditions \eqref{prep_offset} and \eqref{exec_offset}. This observation is quite consistent with that of Figs \ref{nofad_rlf} and \ref{nofad_hof}, resembling that channel fading causes delayed handover. From these results it is evident that channel fading plays a crucial role in determining \emph{handover failure} and \emph{handover latency}.

\begin{figure*}[t]
    \centering
    \begin{minipage}[b]{0.3\textwidth}
    \centering
    \includegraphics[page=12,clip, trim=5cm 7.0cm 2.5cm 4cm,scale=.33]{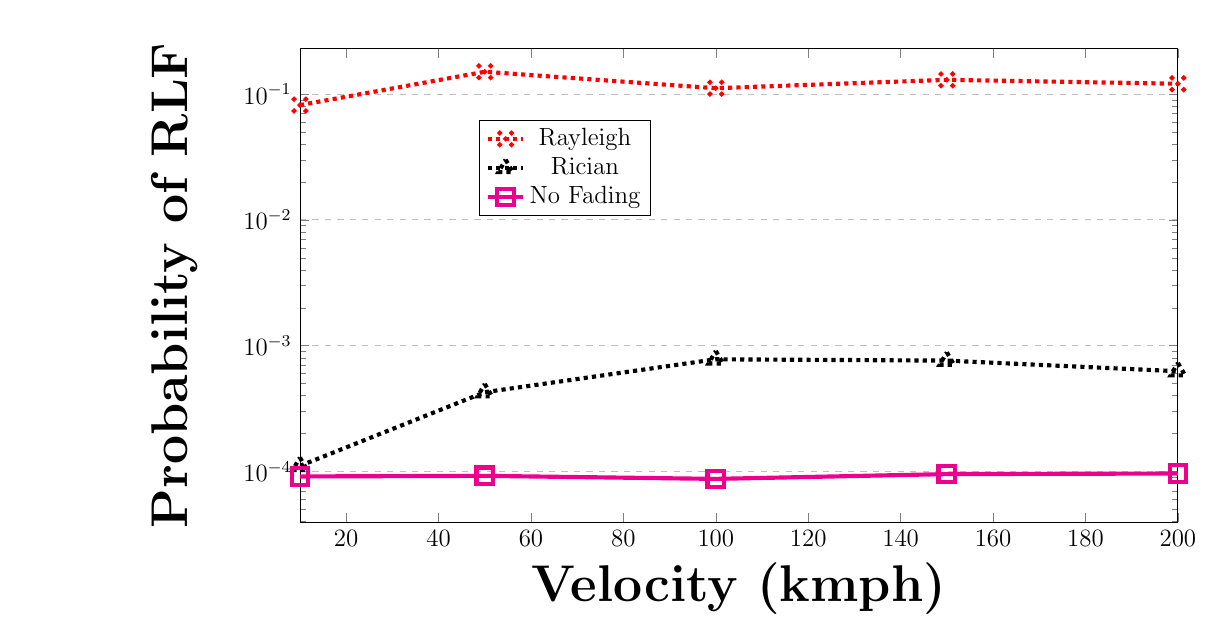}
    \caption{RLF probability vs. UE velocity}
    \label{nofad_rlf}
    \end{minipage}
  \hfill
  \begin{minipage}[b]{0.3\textwidth}
    \centering
    \includegraphics[page=13,clip, trim=5cm 7.0cm 2.5cm 4cm,scale=.33]{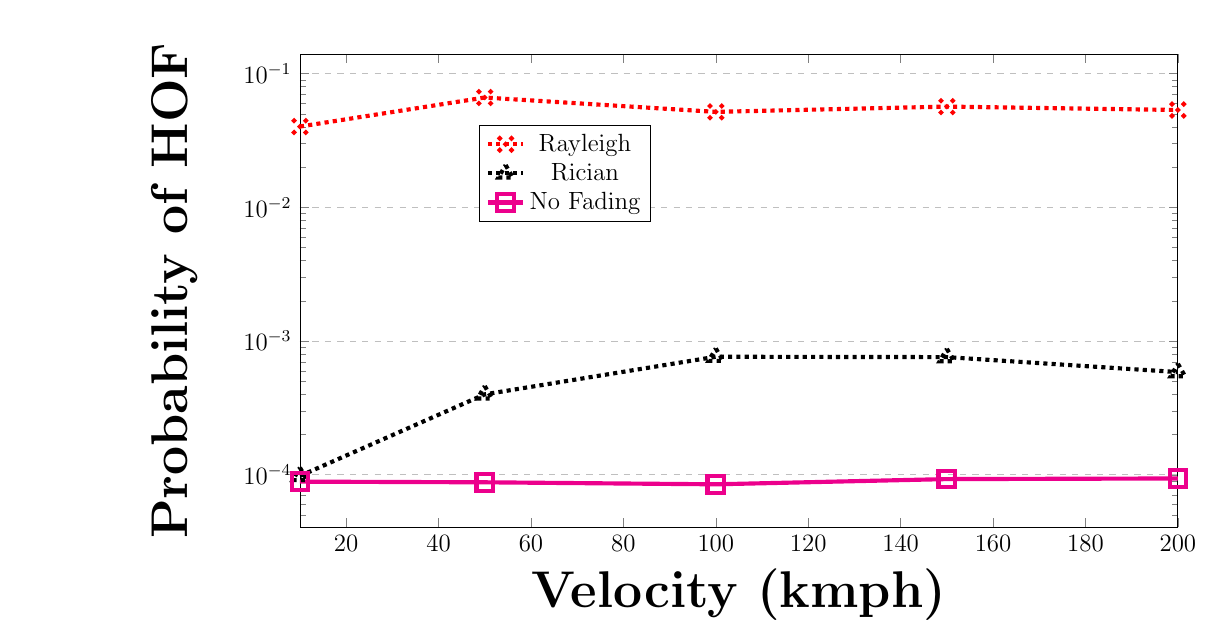}
    \caption{Probability of HOF vs. velocity}
    \label{nofad_hof}
    \end{minipage}
    \hfill
    \begin{minipage}[b]{0.3\textwidth}
         \centering
    \includegraphics[page=9,clip, trim=5cm 3.0cm 2.5cm 4cm,scale=.3]{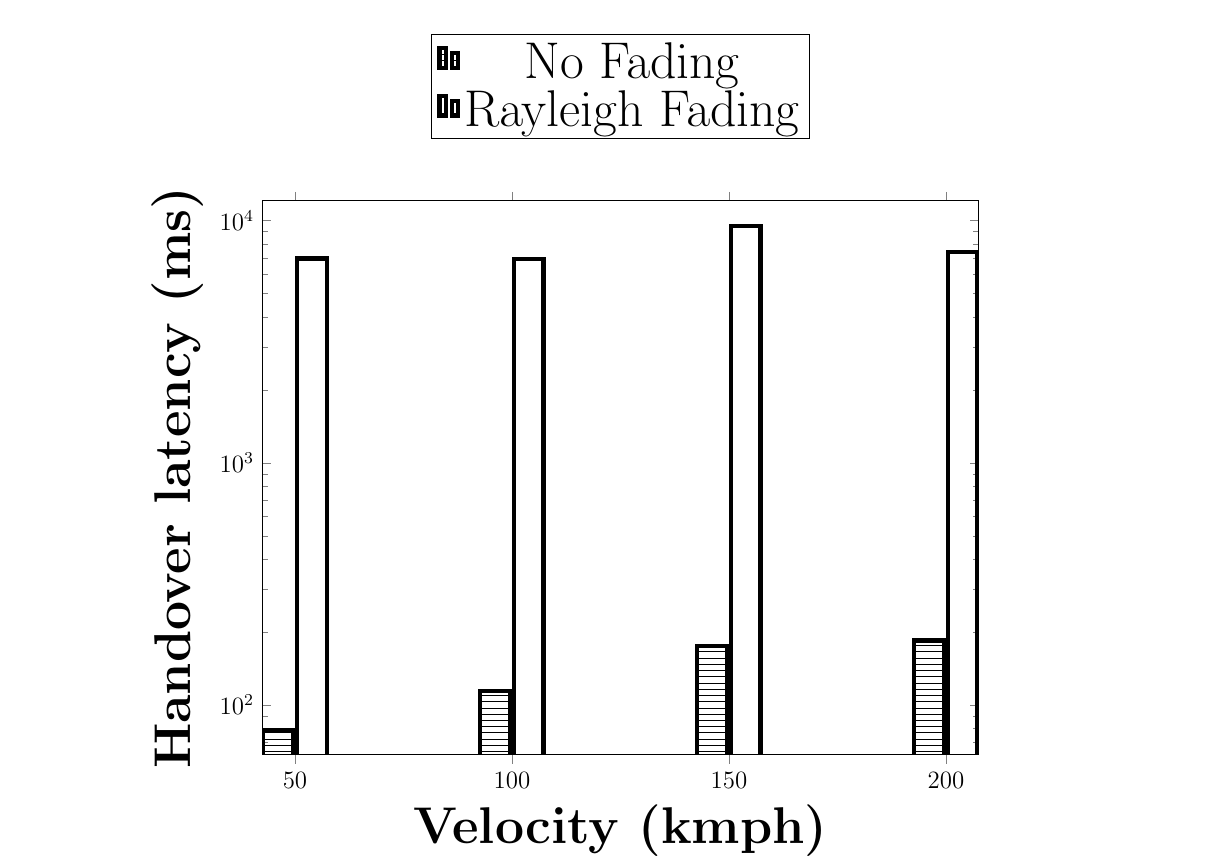}
    \caption{Handover latency vs. velocity}
    \label{lv}
    \end{minipage}
\end{figure*}

For optimal operation of CHO in terms of minimized handover failure and handover latency, choosing appropriate thresholds for $O_{prep}$, $O_{exec}$, $T_{prep}$ and $T_{exec}$ is extremely important. For example, higher values of $T_{prep}$ and $T_{exec}$ may cause RLF with the serving gNB which may cause subsequent HOF. Similarly, higher values of $O_{prep}$ and $O_{exec}$ can also delay handover. {\bf Hence it is worthy to study handover latency and handover failure as a function of $O_{prep}$, $O_{exec}$, $T_{prep}$, $T_{exec}$, UE velocity and channel fading characteristics. Analytic performance evaluation of conditional handover mechanism considering all the aforesaid parameters is unprecedented in literature.} The existing studies on performance analysis are simulation based. 
In \cite{bsics_of_cho}, mobility, outage and signaling aspects of CHO in new radio (NR) systems have been explored through system level simulations. Here CHO has been shown to enhance robustness and outage gain; however improper choice of mobility parameters result in a notable increase in signaling overhead. In \cite{lee2020prediction}, a deep-learning based predictive CHO mechanism has been proposed to enhance early preparation success rates; and to reduce signaling overhead. In \cite{cho_5g_principal}, CHO has been shown to enhance the reliability of handovers (i.e., less handover failure) in 5G systems, where the impact of interference, mobility and signal quality have been explicitly considered. In \cite{10294030}, handover delay in CHO has been reduced by introducing a beam-specific measurement reporting to update Contention Free Random Access (CFRA) resources in 5G systems. The topology-aware skipping approach proposed in \cite{arshad2016handover} aims to reduce unnecessary handovers in CHO by exploiting the topological relationship between gNBs assuming static network topology. This approach may not perform optimally in dynamic network environments. All these existing studies 
\cite{bsics_of_cho}, \cite{lee2020prediction}, \cite{cho_5g_principal}, \cite{10294030}, \cite{arshad2016handover} are simulation based, and do not quantify the handover performance in terms of the mobility management parameters such as $O_{exec}$, $O_{prep}$, $T_{exec}$, $T_{prep}$ and UE velocity. Moreover, the fading characteristics of the channel has not been considered. 

In this work, our \emph{objective} is to analyse the handover performance of CHO explicitly accounting the mobility management parameters such as $O_{exec}$, $O_{prep}$, $T_{exec}$, $T_{prep}$, UE velocity and channel fading characteristics. Our \emph{contributions} can be summarized as follows:

\begin{itemize}
    \item We have characterized the CHO mechanism in terms of $O_{exec}$, $O_{prep}$, $T_{exec}$, $T_{prep}$, UE velocity and channel fading characteristics (Raleigh/Rician) based on a Markov model. From the proposed Markov model, the probability of HOF has been computed. 
    \item Extensive system level simulation has been conducted to explore the effect of $O_{exec}$, $O_{prep}$, $T_{exec}$, $T_{prep}$, UE velocity and channel fading characteristics on RLF, HOF, handover latency and handover packet loss. Moreover, results obtained from the proposed Markov model characterizing CHO has been validated against  simulation results. A special case of the Markov model resembling A3 handover (when handover execution follows handover preparation deterministically) has been validated against ns-3 simulation as well. 
\end{itemize}

\noindent {\bf Our study reveal that optimal configuration of $O_{exec}$, $O_{prep}$, $T_{exec}$ and $T_{prep}$ is actually conditional on underlying UE velocity and fading characteristics. This study will be helpful for the mobile operators to choose appropriate thresholds of the aforesaid mobility parameters under different channel condition and UE velocities.}

The rest of the manuscript is organized as follows. In section \ref{rwork}, related works has been discussed. In section \ref{sysmodel}, considered system model has been described. In section \ref{markov}, we propose the Markov model representation to characterize CHO. In section \ref{res}, both analytic and simulation results have been presented. Finally, section \ref{con} concludes the paper.

\section{Related works}\label{rwork}
In the traditional A3 condition based handover, the handover process is initiated when the following condition is met \cite{SON_handover_parameters}:
\begin{equation}\label{a3handover}
   P_{2}(t') > P_{1}(t')+Hys,
\end{equation}
for a predefined time period called time to trigger (TTT).
Here $P_1(t')$ and $P_2(t')$ are the RSRP values of the serving gNB and the target gNB respectively. Here $Hys$ denotes the hysteresis parameter. The handover is executed immediately after the above condition is satisfied. However, mmWave communication is highly susceptible to environmental losses (e.g., path loss due to high frequency, presence of obstacle etc.). Hence, the channel condition of the target gNB may change right after the handover preparation leading to failure in execution causing HOF. In CHO, the UE does not immediately execute handover to a prepared target cell. A handover takes place only after the UE is certain of the channel condition of the target gNB and thereby avoiding handover to a wrong cell. Therefore, CHO is better suited for 5G systems.

In this section, we undertake an exhaustive review and examination of pertinent research works centered around mobility management and handover performance optimization in the context of CHO.   

\subsection{Mobility management aspect of CHO}
Mobility management involves protocols and procedures in a network to facilitate the smooth movement of UEs. It is crucial for maintaining uninterrupted connectivity as the UEs switch across different access networks.
In \cite{10078238}, an enhanced mobility management mechanism has been proposed \textemdash specifically for multi panel user equipment  and hand blockage scenarios. The study used a simplified hand blockage model, focused on specific mobility scenarios, lacked real network data and emphasized on the need for further investigation in realistic scenarios. In \cite{bsics_of_cho}, authors have explored the mobility, outage and signaling aspects of CHO in NR systems, in order to  enhance robustness and outage gain. It has also been shown that improper system configuration may result in a notable increase in signaling overhead. In \cite{lee2020prediction}, a deep learning based predictive conditional handover (PCHO) has been proposed for 5G networks, aiming to enhance early preparation success rates. Despite the promise, challenges in deep-learning model implementation require careful consideration of training data quality and validation across diverse network scenarios. In \cite{8445882}, statistical blockage models have been employed to capture the impact of the hand, human body, vehicles, and to study the time scales at which mmWave signals are disrupted by blockage. \cite{iqbal2022analysis} examined the mobility performance of two multi panel UE schemes in multi-beam 5G networks \textemdash aiming to compare them with traditional UEs, optimize mobility parameters, and thereby providing insights into handover models and beam management procedures.  In \cite{cho_5g_principal}, authors addressed the challenge of unreliable handover in 5G networks, particularly at higher frequency bands, due to factors like interference, mobility and signal quality. It proposes utilization of CHO as a solution to enhance the reliability of handovers in 5G networks. However, this solution is specific to FR2. In \cite{10294030}, a beam-specific measurement reporting method has been proposed to optimize the performance of contention free random access (CFRA) during CHO in 5G networks. However, limitations arise from the exclusive focus on CFRA optimization during CHO, overlooking broader random access aspects and potential applicability constraints in diverse network scenarios. These studies rely on simulations in controlled environments, potentially limiting real-world applicability. 

\subsection{Handover performance optimization of CHO}
In the rapidly evolving landscape of 5G networks, optimizing the efficiency and reliability of handovers is crucial. Towards that end, separating handover preparation and execution in CHO emerges as key strategy. Researchers prioritize signal strength, random access and resource optimization to enhance 5G network mobility, emphasizing shared concerns about data quality and signaling balance. For instance, \cite{amirijoo20143gpp} proposed an approach for the self-optimization of the random access channel (RACH) which is used to mitigate interruptions during handovers and enhance reliability. In this case, effectiveness could be constrained by the complexity and variability of real-world network conditions. Similarly \cite{stanczak2023conditional} presented an approach which involves beam-specific measurement reporting that can lead to contention free random access (CFRA) resource updating prior to CHO execution. This approach is aimed at maximizing the use of CFRA during mobility, leading to faster and successful handover. It requires additional signaling and coordination among cells, potentially increasing overhead and resource consumption. Moreover, the topology-aware skipping approach proposed in \cite{arshad2016handover} aims to enhance CHO efficiency in cellular networks by intelligently skipping unnecessary handovers based on the topological relationship between neighbouring eNBs. This approach may not perform optimally in dynamic network environments where the assumed static network topology may not accurately reflect real-time conditions. The proposed adaptive UE cell clustering scheme in \cite{joud2018user} enhances CHO efficiency by adjusting cell configurations based on real-time mobility patterns and ensuring seamless connectivity during transitions. However, the effectiveness of this approach hinges on accurate mobility state estimation. So, any inaccuracies or delays in this estimation may lead to non-ideal handover decisions. 

In \cite{9221410}, authors addressed the rising HOFs in urban 5G areas. A Machine Learning (ML) based beam-specific algorithm has been proposed to predict and initiate handovers before RLF occurs; aiming for a reduction in HOF rates and enhanced quality of experience (QoE). However, the accuracy of beam measurements, unforeseen mobility, UE diversity and energy consumption impact practical applicability. \cite{manalastas2023machine} addressed inter-frequency handover failures in 5G systems by introducing an accurate ML model leveraging RSRP data of base stations and interferes \textemdash aiming to improve handover reliability. However, there are limitations in reducing inter-frequency handover failures, including concerns about model generalization, challenges in real-world deployment and sensitivity to dynamic network conditions. In \cite{gharsallah2019sdn}, a software defined handover solution leveraging software defined networking (SDN) has been proposed to optimize handover in ultra dense 5G networks, resulting in reduced handover failure and handover delay. \cite{wang2023investigating} investigates the impact of TTT on handover performance in ultra-dense 5G networks. 

The existing simulation based studies on CHO have not evaluated handover latency and handover failure as a function of mobility parameters such as $O_{exec}$, $O_{prep}$, $T_{exec}$, $T_{prep}$, UE velocity and channel fading characteristics. In this study, an effort has been initiated to explore handover latency and handover failure considering the aforesaid system parameters.

\section{System model and assumptions} \label{sysmodel}
We consider a 5G network with $B$ number of gNBs providing ubiquitous coverage to $N$ UEs. The set of gNBs is denoted by $\mathcal{B}=\{1,\,2,\cdots,\,B\}$. The set of UEs is denoted by $\mathcal{U}=\{1,\,2,\cdots,\,N\}$. The layer 3 (L3) RSRP measurements are used for all handover related decisions \cite{iqbal2022analysis}. In the subsequent subsections, we describe the channel model of gNB$\rightarrow$UE channels. Then we describe the RLF and HOF models. All the notations used in this study are summarized in Table \ref{symtab}. 

\begin{table}[h]
    \centering
    \caption{Symbol table}
    \begin{tabular}{|c|c|}
    \hline
         \textbf{Symbol}  & \textbf{Meaning} \\
         \hline
         $P_1(t)$ & RSRP value at time $t$ from gNB $1$.\\
         \hline
         $P_{2}(t)$ & RSRP value at time $t$ from gNB $2$.\\
         \hline
         $T_{prep}$ & Time for handover preparation.\\
         \hline
         $O_{prep}$ & Power offset for handover preparation.\\
         \hline
         $T_{exec}$ & Time for handover execution.\\
         \hline
         $O_{exec}$ & Power offset for handover execution.\\
         \hline
         $\mathcal{B}$ & Set of gNBs.\\
         \hline
         $\mathcal{U}$ & Set of UEs.\\
         \hline
         $P^T_c$ & Transmit power of the gNB $c$.\\
         \hline
         $P^r_{ci}$ & Received power of the gNB $c$ at UE $i$. \\
         \hline
         $\mathbf{h}_{ci}$ & Channel Coefficient Vector from gNB $c$ to UE $i$. \\
         \hline
         $\mathcal{K}$ & Rician factor. \\
         \hline
         $\gamma_i$ & Signal to Interference and noise ratio (SINR) of UE $i$. \\
         \hline
         $RLF$ & Radio Link Failure. \\
         \hline
         $HOF$ & Handover Failure. \\
         \hline
         $P$ & Transition Matrix. \\
         \hline
         $p(i,\;j)$ & Transition probability from state $i$ to $j$. \\
         \hline
         $\pi_{HOF}$ & Probability HOF.\\
         \hline
    \end{tabular}
    \label{symtab}
\end{table}

\subsection{Channel model}
Let at time $t$, UE $i$ be served by gNB $c\in\mathcal{B}$ with the transmit power $P^T_c$. Let UE $i$ is at a distance $d_{ci}$ from gNB $c$. The received power $P^r_{ci}(t)$ at UE $i$ from gNB $c$ is computed as follows \cite{channel_coeff}:
\begin{equation} \label{pathloss}
    P^r_{ci}(t)=|\mathbf{h}^H_{ci}\sqrt{d_{ci}^{-\alpha}}|^2P^T_c.
\end{equation}
\noindent Here $\mathbf{h}_{ci}\in \mathbb{C}^{M\times 1}$ is the channel coefficient vector for the gNB $\rightarrow$ UE channel and $\alpha$ is the pathloss exponent. The value of the $l^{th}$ component of the channel coefficient vector for Rician fading, denoted by $h^l_{ci}$ is given by a complex Gaussian distribution $\mathcal{CN}(\mu,\sigma^2)$ \cite{rician_fading}. Denoting by $\mathcal{K}$ as the Rician factor, the mean $\mu$ and standard deviation $\sigma$ of the complex Gaussian is given by $\mu=\sqrt{\frac{\mathcal{K}}{2(\mathcal{K}+1)}}$ and $\sigma=\sqrt{\frac{1}{2(\mathcal{K}+1)}}$ respectively. For $\mathcal{K}=0$, $\mathbf{h}_{ci}$ represents Rayleigh fading.

The signal-to-interference plus noise ratio (SINR) $\gamma_i$ measured by UE $i$ when connected to gNB $c$ is given as below:
\begin{equation}
    \gamma_i=\frac{P^r_{ci}(t)}{\displaystyle\sum_{j\ne c\, j\in\mathcal{B}}|\mathbf{h}^H_{ji}\sqrt{d_{ji}^{-\alpha}}|^2P^T_j+\omega^2}
\end{equation}
\noindent Here $\omega^2$ is additive white Gaussian noise (AWGN) power at the UE.

\subsection{RLF and HOF}
Although the detailed implementation of RLF varies from chipset-to-chipset \cite{share}, in this work, we stick with the widely used and accepted SINR based definition of RLF 
\cite{survey1}, \cite{rlfdef} \cite{lopez}:\\
During communication, $\gamma_i$ may drop below a predefined threshold $\gamma_{out}$ by the mobile operator for $N310$ consecutive frames. In such a case,  $T310$ timer (also known as RLF timer) is started. During the RLF timer, if the SINR exceed another threshold $\gamma_{in} (>\gamma_{out}$) known as in-synch event, the ongoing RLF event is cancelled. Otherwise, RLF is declared as soon as the RLF timer expires.  Accordingly, the connection of the UE with the serving gNB is lost. \\
The occurrence of RLF during a handover process is one of the prime causes of handover failures \cite{lopez}. In CHO, for the handover preparation phase to complete, the RSRP at the UE must satisfy condition \eqref{prep_offset} for all sampling instances during $T_{prep}$. Similarly, after handover preparations are done, condition \eqref{exec_offset} must be satisfied for all sampling instances during $T_{exec}$ for the handover execution phase to be completed.  Violation of these conditions due to fading results in restarting the corresponding phase, leading to delayed handover (shown in Fig. \ref{cho_mech}). In case of delayed handover, $\gamma_i$ may keep on decreasing which may cause RLF  during handover preparation/execution phase. Such a RLF causes HOF as well as shown in Fig. \ref{rlf_e}.

\begin{figure}[h]
    \centering
    \includegraphics[scale=0.4]{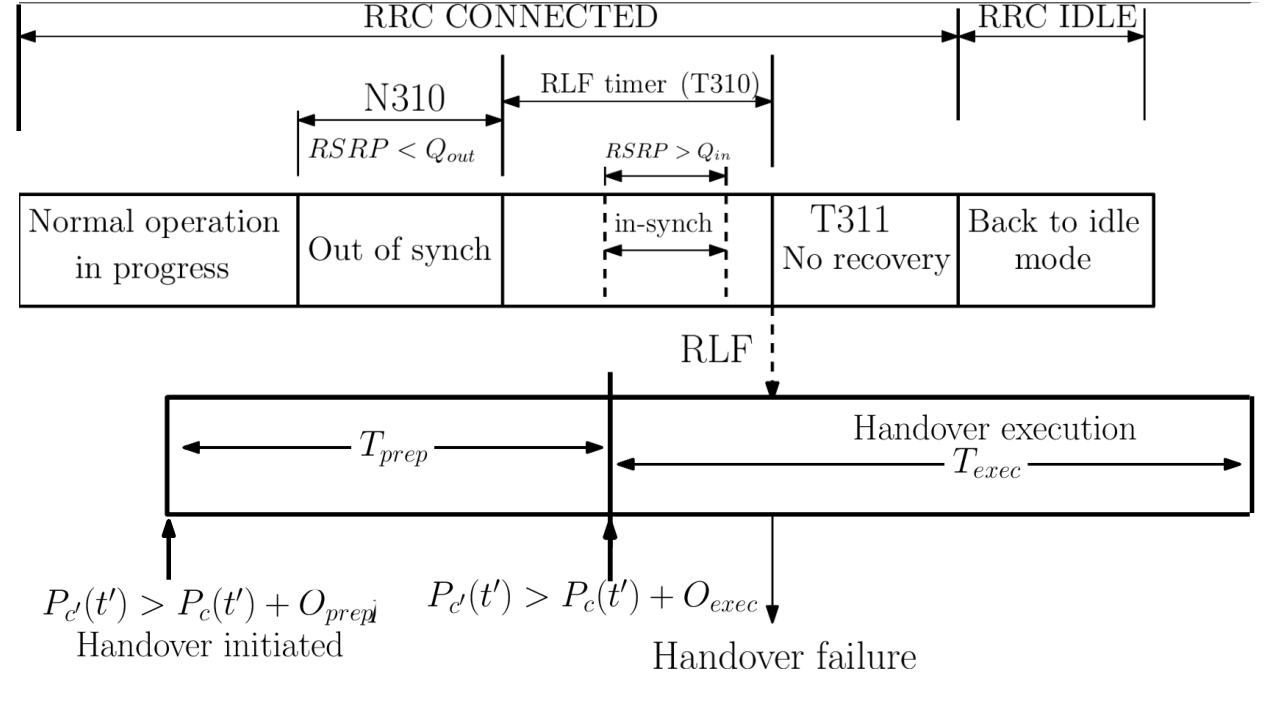}
    \caption{RLF during handover execution}
    \label{rlf_e}
\end{figure}

\section{Performance model for conditional handover} \label{markov}

In order to analyze the performance of CHO, we model the different stages of a CHO process as a discrete-time Markov chain. The Markov model considers the effect of $O_{exec}$, $O_{prep}$, $T_{exec}$, $T_{prep}$, UE velocity and channel fading characteristics. In the subsequent subsections, we describe the states of the Markov chain, computation of transition probabilities, and finally probabilities of HOF. 

In NR systems, the UE measures the RSRP using the synchronization signal block (SSB) bursts transmitted by each gNB. The SSB periodicity is $20$ ms \cite{iqbal2023analysis}.  Therefore, sampling interval between two successive RSRP measurements has been considered to be $20$ ms, i.e., $T_{sample}=20$ ms. We assume that time is discretized into slots where each slot has a duration equal to $T_{sample}$. In each state of the Markov chain,  RSRP from all gNBs are measured at the beginning of each time slot. Based on these measurements, the UE transitions into a new state from its current state. The resulting Markov model is shown in Fig. \ref{cho1}.

\subsection{States of the Markov chain:}
In this section, we define different states of the Markov chain. Let us consider that $S$ represents the set of all states of the Markov chain. Elements of the set $S$ are defined below:

\begin{figure*}[t]
    \centering
    \includegraphics[scale=0.8]{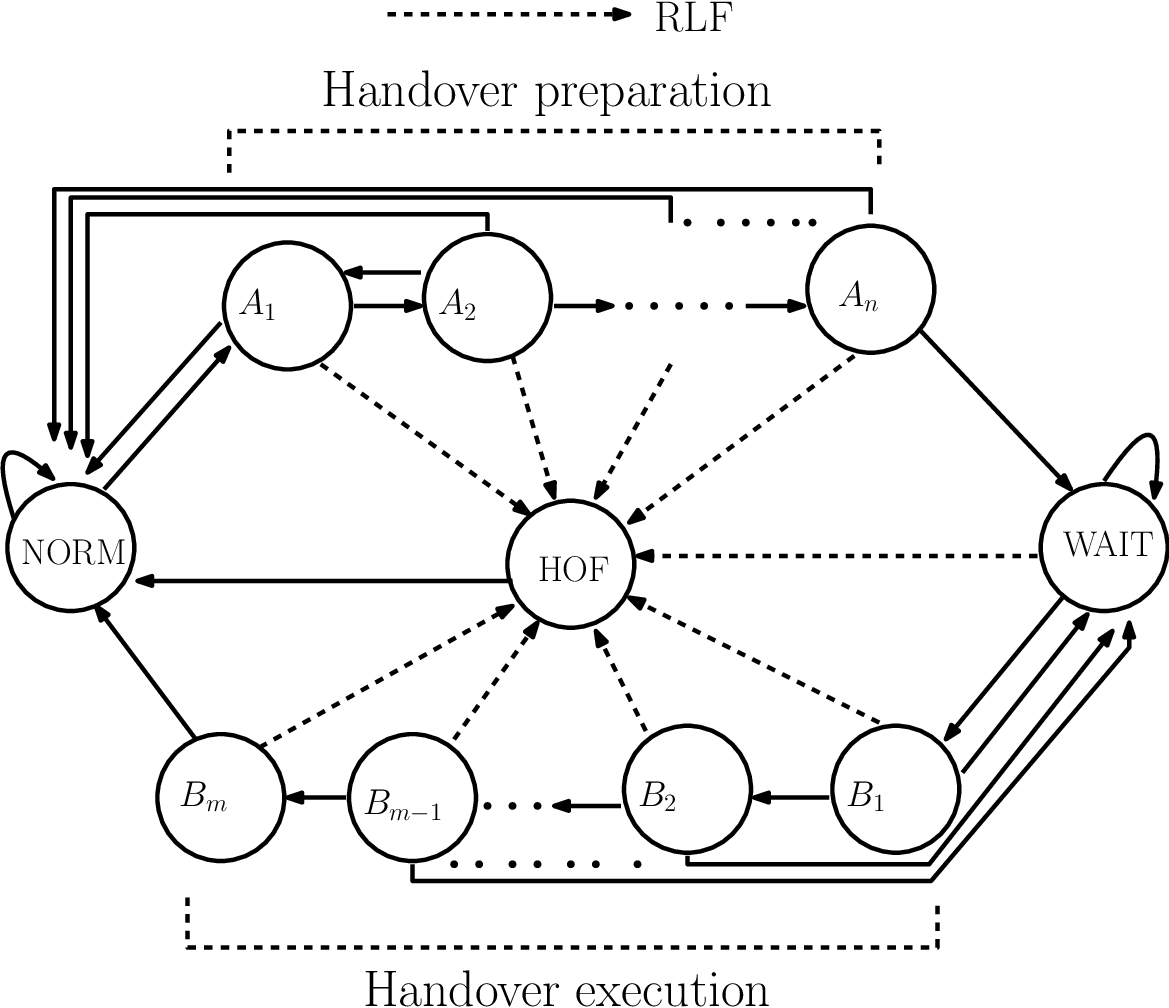}
    \caption{Markov chain characterizing conditional handover process}
    \label{cho1}
\end{figure*}

\textbf{State $NORM$}:
State $NORM$ represents normal communication with the gNB, i,e., there is no underlying RLF or handover. 

\textbf{States $A_1$ to $A_n$}: After the handover has been initiated at some time $t'$, the handover preparation process begins. For the handover preparation process to be completed, the RSRP at the UE must satisfy condition \eqref{prep_offset}  for all sampling instances during $T_{prep}$. During this process, the UE sends RSRP measurements to the current gNB at discrete time steps. Hence, we represent the handover preparation process, using $n$ discrete states. Here $n=\lceil \frac{T_{prep}}{T_{sample}} \rceil$, where $T_{sample}$ represents the sampling interval. The handover preparation phase is completed at $A_n$.

Given any state $A_i$, where $1\le i \le n$ there is a transition to state $A_{i+1}$ if and only if condition \eqref{prep_offset} is satisfied. When the state $WAIT$ is reached, it signifies that the handover preparations are done, and the UE is waiting for handover execution. In case,  condition \eqref{prep_offset} is violated in any of the states from $A_1$ to $A_n$, there is a transition back to state $NORM$.  

\textbf{State $WAIT$}: $WAIT$ represents the state that handover preparation has been completed and the UE is waiting for the beginning of the handover execution phase. The Markov chain deterministically switches from $A_n$ to $WAIT$ in the next sampling instance.

\textbf{States $B_1$ to $B_m$}: States $B_1$ to $B_m$ constitutes handover execution phase. Here $m=\lceil \frac{T_{exec}}{T_{sample}} \rceil$. After the UE reaches the state $A_m$, the handover must now be executed. For the handover execution to completed, the UE must satisfy condition $P_2(t')> P_1(t')+O_{exec}$ in every sampling instance during $T_{exec}$. Just like the handover preparation phase, we represent the handover execution phase $m$ discrete states: $B_1, B_2, \ldots B_m$. 

Given any state $B_j$, where $1\le j \le m$ there is a transition to state $B_{j+1}$ if and only if condition \eqref{exec_offset} is satisfied. When the state $B_{m}$ is reached, it signifies that the UE is ready to handover to the target gNB. If condition \eqref{exec_offset} is violated when the UE is in any of the states $B_1, B_2, \ldots,B_m$, there is a transition back to state $WAIT$ to restart the handover execution process again. From $B_m$ a transition happens to the state $NORM$ if and only if there is no RLF, and the UE has undergone a successful handover.

Thus $A_i \rightarrow NORM$ ($1 \leq i \leq n$) and $B_j \rightarrow WAIT$ ($1 \leq i \leq m$) transitions capture the effect of channel fading characteristics on CHO process.

\textbf{State $HOF$:} This state is reached when a HOF occurs due to RLF during the handover preparation or execution phase. Hence, there is a nonzero probability of transition to this state from any of the previous states. After the occurrence of an RLF the UE connects to the nearest gNB and normal communication resumes again. This is signified by a transition from $HOF$ to $NORM$ with probability one.

\subsection{Computing transition probabilities:}
The transition probability matrix $P\in \mathbb{R}^{(n+m+3)\times (n+m+3)}$ is a function of UE velocity and channel fading characteristics. Here $p(i,j)$ is the $(i,j)$th entry of the matrix $\mathbb{R}$, representing the probability of transitions from state $i$ to state $j$. $A_i \rightarrow NORM$ ($1 \leq i \leq n$) and $B_j \rightarrow WAIT$ ($1 \leq i \leq m$) transitions capture the effect of channel fading characteristics, whereas the UE velocity determines the correlation between the states $A_i$s and $B_j$s. To keep the model simple and tractable, in this study, $p(i, j)$ has been computed from simulation traces as follows:

\begin{equation} \label{transition}
    p(i,j)=\frac{Number\; of\; transitions \; from\; state\; i \; to\; state\; j}{N},
\end{equation}

\noindent where $N$ is the total number of states observed; It may be noted that $p(i, j)$ is a function of UE velocity and channel fading characteristics (Raleigh/Rician). 

\subsection{Computing the steady state probability}
In this subsection, we model HOF probability as the steady state probability of residing the Markov chain in state HOF. Subsequently, we first prove that the proposed finite state Markov chain is \emph{irreducible} and \emph{aperiodic}. Then, solving the \emph{Chapman-Kolmogorov} equations, we compute the steady state probability of residing in the state HOF, i.e, the probability of handover failure. 

\begin{lemma}
The Markov chain defined by $P$ is aperiodic and irreducible, if both $r_{ss} (k)$ and $r_{ss} (k+1)$
are non-zero for all $s \in S$, where $r_{ss} (k)$ is the probability that starting from state $s$, the Markov process transitions back to state $s$ in exactly $k$ steps. 
\end{lemma}
\begin{proof}
As per definition, Markov chain defined by $P$ is said to be aperiodic if, the greatest common divisor (GCD) of the path lengths of $s \rightarrow s$ paths is $1$ for all $s \in S$ \cite{markov_book}.  

For state $s'$ ($s' \in S \setminus \{NORM,\;B_m,\;HOF\}$), we have $r_{s's'}(k)>0$ and $r_{s's'}(k+1)>0$ making the gcd of all the path lengths transitioning back to $s'$ as $1$. This happens because, when the UE is in any arbitrary state $s'$, an RLF causing a HOF can occur resulting in a transition to the state $HOF$, and ALL $s'$ is reachable from $NORM$. Let us assume that there is a  $s' \rightarrow s' $ path of length $k$ via $s_1$ and $s_2$,  where $s_1, s_2 \in S \setminus \{NORM,\;B_m,\;HOF\}$; and length of the  $s_1 \rightarrow s_2$ path is $1$. HOF is reachable from both $s_1$ and $s_2$. Hence, there exist a path of length $k-1$ from $s' \rightarrow s'$ via $s_1$ excluding $s_2$, i.e., the path $s' \rightarrow s_1 \rightarrow HOF \rightarrow NORM \rightarrow s'$. 

There exist two paths of transition from $NORM\rightarrow NORM$ as follows: (1) $NORM\rightarrow A_1\rightarrow HOF \rightarrow NORM$ with path-length $3$, and (2) $NORM\rightarrow A_1 \rightarrow A_2\rightarrow HOF \rightarrow NORM$ with path-length $4$. Therefore, $r_{NORM\;NORM}(3)>0$ and $r_{NORM\;NORM}(4)>0$ making the GCD of length of all paths from $NORM\rightarrow NORM$ to be $1$.

There exist two paths of transition from $HOF\rightarrow HOF$ as follows: (1) $HOF\rightarrow NORM \rightarrow A_1\rightarrow HOF$ with pathlength $3$, and (2) $HOF\rightarrow NORM \rightarrow A_1 \rightarrow A_2\rightarrow HOF$ with pathlength $4$. Therefore, $r_{HOF\;HOF}(3)>0$ and $r_{HOF\;HOF}(4)>0$ making the GCD of length of all paths from $HOF\rightarrow HOF$ to be $1$.

There exist two paths of transition from $B_m\rightarrow B_m$ as follows: (1) $B_m\rightarrow NORM\rightarrow A_1\;\ldots, A_n\rightarrow WAIT\rightarrow B_1\;\ldots B_m$ with path length $n+m+1$, and (2) $B_m\rightarrow HOF\rightarrow NORM\rightarrow A_1\;\ldots, A_n\rightarrow WAIT\rightarrow B_1\;\ldots B_m$ with path length $m+n+2$. Therefore, $r_{B_m}\;B_m(n+m+1)>0$ and $r_{B_{m} B_{m}}(n+m+2)>0$ making the GCD of length of all paths from $B_m\rightarrow B_m$ to be $1$.

Therefore, the GCD of the path-lengths $s \rightarrow s$ for all $s \in S$ is $1$. Hence we conclude that the Markov chain defined by $P$ is \emph{aperiodic}. Also, all the states are reachable from $NORM$, and $NORM$ is reachable from all other states. Therefore, there exists a $k$ such that  $r_{ss'}(k)>0$ for all $s$ and $s'$ belonging to $S$. Hence, the considered Markov chain is \emph{irreducible} as well. 
\end{proof}

In the next subsection, we compute $\pi_{HOF}$ the steady state probability of remaining in state HOF by solving Chapman-Kolmogorov equation.

\subsection{Computing $\pi_{HOF}$}
Let the steady-state probability vector of the above-defined Markov chain be $\mathbf{\pi}\in [0,1]^{1 \times (n+m+3)} $. Then, the Chapman–Kolmogorov equations are:
 \begin{eqnarray}
  \mathbf{\pi}=\mathbf{\pi}P, \\
 \displaystyle\sum_{s \in S}\pi_s=1.    
 \end{eqnarray}
By solving this system of linear equations, we can find the steady-state probabilities as follows:
\begin{equation}\label{solution3}
     \pi_{HOF}=1-\phi \times \pi_{NORM}-\delta \times \pi_{WAIT}
\end{equation}
\noindent where
\begin{eqnarray*}
    \pi_{NORM}=\frac{1}{\phi -\alpha + (\frac{p(A_n, WAIT)a_{n}}{\beta})(p_{B_{m} NORM}b_{m}-\delta)}\label{solution1}\\
    \pi_{WAIT}=\frac{1}{(\alpha-\phi)(\frac{\beta}{p(A_n, WAIT)a_{n}})-p(B_m, NORM) b_{m}+\delta},\label{solution2}\\
     a_{i}=p(A_{i-1}, A_{i}) p(A_{i-2},A_{i-1})\cdots p(NORM ,A_1),\\
     b_{i}=p(B_{i-1},B_i)p(B_{i-2},B_{i-1})\cdots p(WAIT, B_1),\\
     \alpha= p (NORM, NORM)+\displaystyle\sum_{i=2}^{n}a_ip(A_i, NORM)-1,\\
     \beta= p (WAIT, WAIT)-1+\displaystyle\sum_{i=1}^{m-1}b_ip (B_i, WAIT),\\
      \phi=1+ a_{NORM} + \displaystyle\sum_{i=2}^{n}a_i,\\
       \delta=1+\displaystyle\sum_{i=1}^{m}b_i,\\
       \lambda=b_{1}.
\end{eqnarray*}

\paragraph*{Modelling A3 based handover}: The traditional A3 condition based handover mechanism can be modeled as an special case of the proposed Markov chain characterizing CHO mechanism. The proposed Markov chain can be modified to model  A3 handover mechanism by adjusting the mobility parameters in one of the following two ways: (1) $O_{prep}=Hys, \;O_{exec}=0,\; T_{prep}=TTT,\; T_{exec}=0$ OR (2) $O_{prep}=0,\; O_{exec}=TTT,\; T_{prep}=0,\; T_{exec}=TTT$. In the first case, the states $B_i$ for all $1\leq i\leq m$ cease to exist, and there is direct transition from $WAIT\rightarrow NORM$ with probability $1$. That means the preparation event boils down to the A3 event, and the handover is executed immediately after the A3 condition is satisfied. On the other hand, for the second case, the states $A_1$ to $A_n$ cease to exist, and the transition $NORM \rightarrow WAIT$ happens with probability $1$. Here condition \eqref{exec_offset} for handover execution boils down to the A3 condition, i.e., condition \eqref{a3handover}. Accordingly, probability of handover failure for A3 handover can be computed from equation \eqref{solution3}.

In the next section, we evaluate the effect of different mobility parameters on CHO performance in terms of handover latency, handover packet loss and handover failure.

\section{Results and discussions}  \label{res}
In this section, we evaluate the effect of different mobility parameters such as $O_{exec}$, $O_{prep}$, $T_{exec}$, $T_{prep}$ and UE velocity on handover performance of CHO under different fading scenarios. We have considered probability of radio link failure, probability of handover failure, handover latency and handover packet loss rate as performance evaluation metrics. Results obtained from the proposed analytic framework has been validated against simulation results as well. Moreover, an special case of the proposed Markov model resembling A3 handover has been validated against ns-3 simulations.

\subsection{Base-case validation using ns-3}
The A3-based handover can be visualized as a special case of the CHO where: $O_{prep}=Hys, \;O_{exec}=0,\; T_{prep}=TTT,\; T_{exec}=0$; the states $B_i$ for all $1\leq i\leq m$ cease to exist, and there is direct transition from $WAIT\rightarrow NORM$ with probability $1$, i.e., $p(WAIT, B_1)$, $p(B_1, B_2)$, $\ldots$, $p(B_{m-1}, B_m)$, $p(B_m, NORM)$ is set to $0$. On the other hand, $p(WAIT, HOF)$ and $p(HOF, NORM)$ are set to $1$. Accordingly, probability of handover failure for A3 handover is computed from equation \eqref{solution3}.

\begin{table*}[t]
\caption{Parameter set used for the lena-x2-handover-measures scenario in ns-3 for simulation}
    \centering
    \begin{tabular}{|c|c|}
    \hline
    Parameter & Description \\
    \hline
     Mobility & ConstantVelocityMobilityModel (default)\\
     \hline
     Pathloss model & FriisPropagationLossModel with NakagamiPropagationLossModel \\
     \hline
     Transmit power of eNB & 40 dBm\\
     \hline
     Traffic model & Constant Bit Rate (CBR) traffic for data streaming (default). \\ 
     \hline
      Handover algorithm & A3 RSRP Handover Algorithm \\ 
      \hline
      Frequency & $2.12$ Ghz\\
      \hline
      Operating system & Ubuntu 22.04.3 LTS\\
      \hline
      Processor & 12th Gen Intel® Core™ i7-12700\\
      \hline
      Core & 20\\
      \hline
      GPU & AMD® Radeon rx 640\\
      \hline
      RAM & 32 GB\\
      \hline    
    \end{tabular}
    
    \label{ns3_tab}
\end{table*}

\begin{figure}[htbp]
    \centering
    \includegraphics[page=8,clip, trim=5cm 7.0cm 2.5cm 4cm,scale=.36]{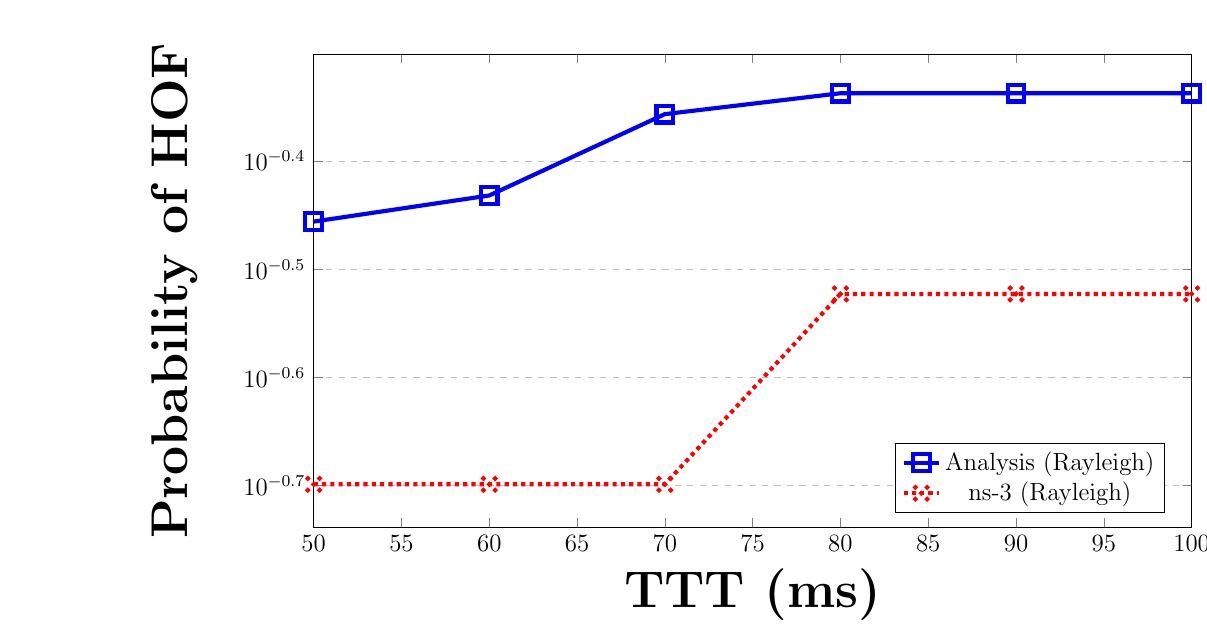}
    \caption{Probability of HOF vs. TTT (Hys=11dB)}
    \label{ns3}
\end{figure}

We set up the following ns-3 experiment to measure handover failure probability: the network scenario consists of three LTE evolved Node Bs (eNBs) placed $500$ meters apart from each other on a straight line (say eNB 1, eNB 2 and eNB 3). Initially, the UE is associated with eNB 1. The UE moves from eNB 1 to eNB 3 through a linear trajectory with a constant velocity of $22$m/s. The ``lena-x2-handover-measures" example of LTE LENA module has been modified to design the experiment. We have added additional trace sinks to monitor and record RLFs, successful handovers and HOFs. Detailed parameter settings used for the ns-3 simulation are depicted in Table \ref{ns3_tab}. {\bf In order to closely resemble the behavior of millimeter waves the carrier frequency has been set to $2.12$ Ghz by modifying the downlink E-UTRA absolute radio frequency channel number (EARFCN).} The downlink and uplink EARFCNs have been set to $100$ and  $18100$ respectively. The EARFCNs have been set using LteEnbNetDevice::GetTypeId() function of lte-end-net-device.cc module.  We have used a customized Python script to run the ns-3 source code for different values of TTT and record occurrence of HOF from the trace values generated by ns-3. The source code along with the Python script has been made available on GitHub 
\footnote{GitHub link: \url{https://github.com/Handover-ns3-works/ns-3-dev-git/blob/master/scratch/handover-ns3-works/rlf_runner.py}}. 

Probability of HOF shows an increasing trend with TTT (shown in Fig. \ref{ns3}). This result has been obtained from $1000$ independent run of the simulator. To compute the transition probabilities, the value of $N$ has been set to $1000$ in equation \eqref{transition}. HOF shows increasing trend with TTT because longer TTT leads to increased RLF. Both analysis and simulation results show similar trend. Difference between analytic and simulation results is due to the fact that the analytic model do not consider the HOF cases due to random access channel (RACH) failure, which is considered by full stack simulator as ns-3.

To the best of authors' knowledge, CHO is not implemented in ns-3. Hence, for further exploration of handover performance of CHO mechanism, we have developed a system level simulator in Python. Source code of the developed simulator is available in \cite{github}. In the next subsection, we first present the simulation set-up. Then we describe both the analytic and simulation results. 

\subsection{Simulation setup}
We have considered $8$ NR gNBs deployed across a $500\times1000$ square meters rectangular area. The gNBs are serving $20$ UEs. The minimum distance between two gNBs is $50$ meters, and the maximum distance is $500$ meters. Initially, the UEs are positioned randomly in the region. During the simulation, the UEs are moving according to random way point mobility model with a uniform random velocity between $0$ km/h and $50$ km/h. The NR eNBs have a carrier frequency of $28$ GHz with a $100$ MHz bandwidth. The transmitting power of gNBs has been set to $40$ dBm \cite{iqbal2023analysis}. The Rician factor has been set to $3$ dB. The received power at UEs have been computed using equation \eqref{pathloss}, where $\alpha=2$. The Gaussian noise figure is $-114$ dBm. In NR systems, the UE measures the RSRP using the synchronization signal block (SSB) bursts transmitted by each gNB. The SSB periodicity is $20$ ms \cite{iqbal2023analysis}.  Therefore, sampling interval between two successive RSRP measurements has been considered to be $20$ ms, i.e., $T_{sample}=20$ ms. Default values of the mobility parameters are as follows \cite{bsics_of_cho} \cite{8445882} \cite{cho_5g_principal}: $O_{exec}$=$6$ dB, $O_{prep}$= $1$ dB, $T_{prep}$=$100$ms and $T_{exec}$=$80$ms. 

\subsection{Simulation results}
\begin{figure}[htbp]
    \centering
    \includegraphics[page=2,clip, trim=5cm 7.0cm 2.5cm 4cm,scale=.36]{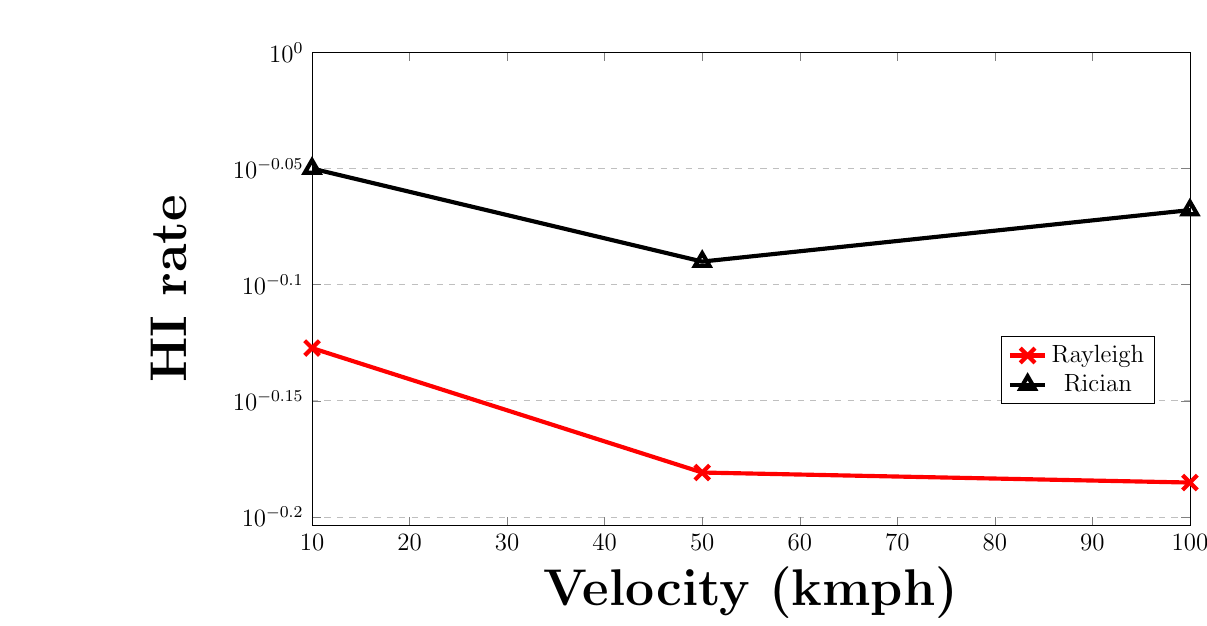}
    \caption{Handover initiation rate (in log scale) vs. velocity [$\mathbf{T_{prep}=100\;ms}$, $\mathbf{T_{exec}=80\;ms}$, $\mathbf{O_{exec}=6}$ dB, $\mathbf{O_{prep}=1}$ dB]}
    \label{hov}
\end{figure}

Fig. \ref{hov} depicts the effect of velocity on the handover initiation (HI) rate. The HI rate is defined as the frequency of entering in handover preparation phase. Here velocity varies uniformly from $10$ kmph to $100$ kmph with a step of $20$ kmph. For the case of Rayleigh fading, we observe a decrease in HI rate as the velocity increases. This is because, with an increase in velocity, the UE leaves the coverage of the serving gNB frequently before condition \ref{prep_offset} is satisfied for $T_{prep}$ period of time. Hence, the HI rate also decreases. In case of Rician fading, the HI rate remains at a stable value for all velocities. This is because the Rician channel has a direct Line of sight (LoS) link between the gNB and the UE. This improves the coverage of the gNB because the RSRP at the UE is stronger in comparison to that for Rayleigh fading case. As a consequence, the UE does not leave the coverage of the serving gNB during the $T_{prep}$ period. The HI rate do not decrease. {\bf From this result, we conclude that handover is initiated more frequently in case of Raleigh fading as compared to that of Rician fading.}   

\begin{figure}[htbp]
    \centering
    \includegraphics[page=1,clip, trim=5cm 7.0cm 2.5cm 4cm,scale=.36]{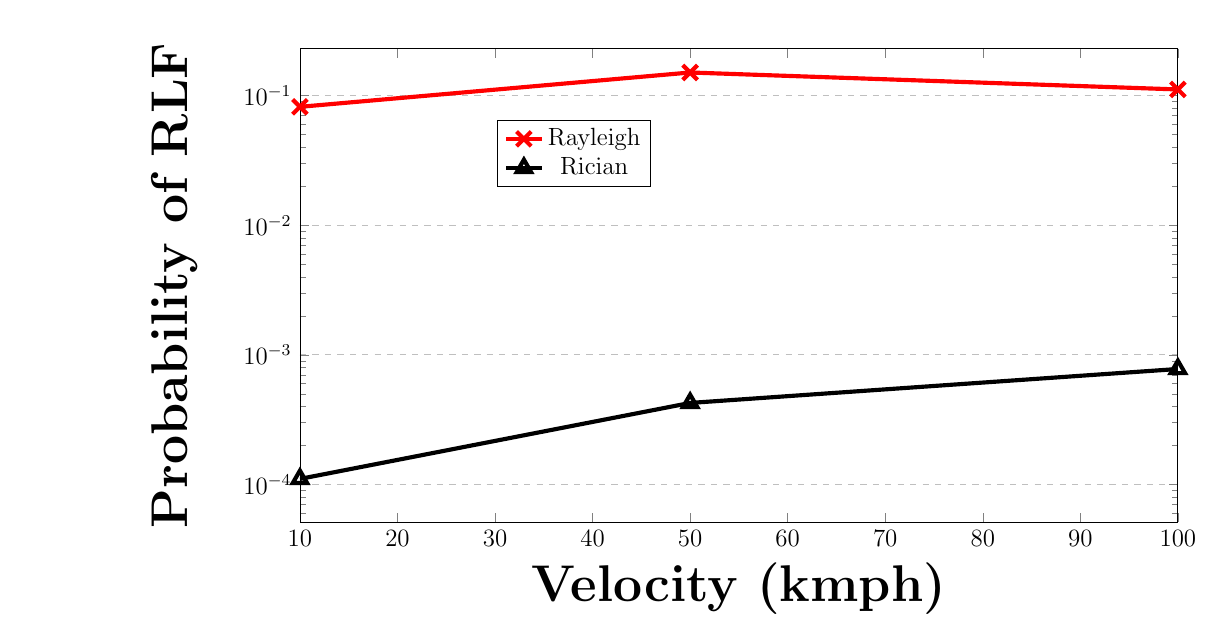}
    \caption{Probability of RLF (in log scale) vs. velocity  [$\mathbf{T_{prep}=100\;ms}$, $\mathbf{T_{exec}=80\;ms}$, $\mathbf{O_{exec}=6}$ dB, $\mathbf{O_{prep}=1}$ dB]}
    \label{rlfv}
\end{figure}

\begin{figure}[htbp]
    \centering
    \includegraphics[page=3,clip, trim=5cm 7.0cm 2.5cm 4cm,scale=.36]{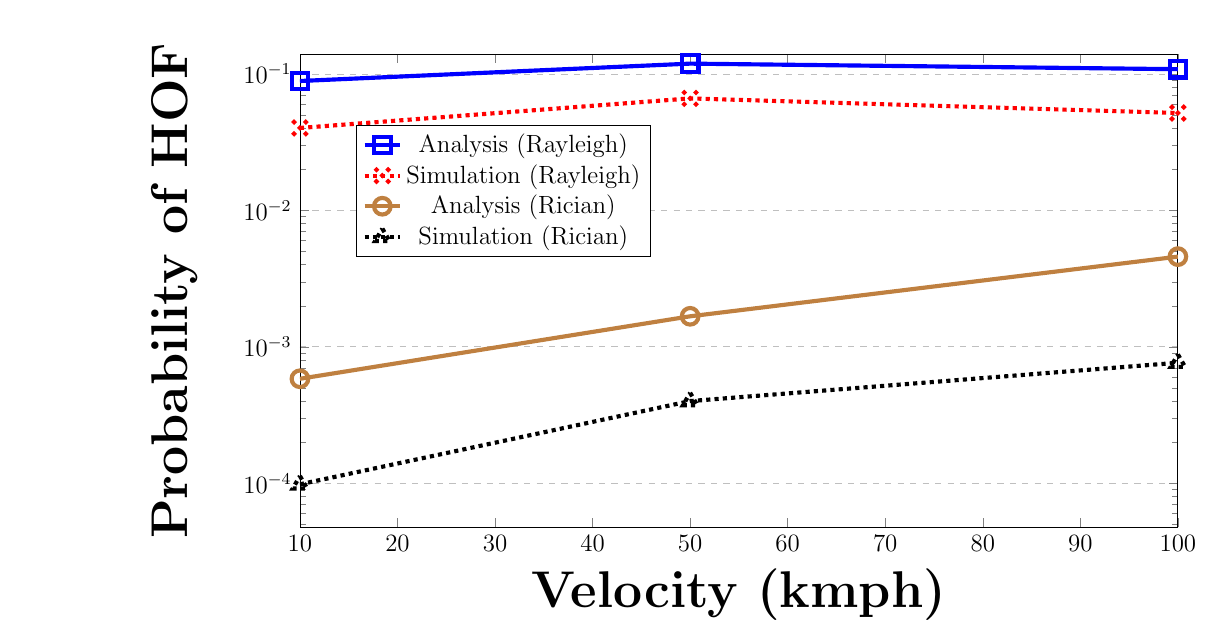}
    \caption{Probability of HOF (in log scale) vs. velocity [$\mathbf{T_{prep}=100\;ms}$, $\mathbf{T_{exec}=80\;ms}$, $\mathbf{O_{exec}=6}$ dB, $\mathbf{O_{prep}=1}$ dB]}
    \label{hofv}
\end{figure}
Fig. \ref{rlfv} depicts the effect of velocity of the UE on the probability of RLF. The velocity varies from $10$ kmph to $100$ kmph with a step of $20$ kmph. For the case of Rayleigh fading, we observe an initial increase in RLF probability as the velocity increases upto $50$ kmph. As the velocity increases  beyond $50$ kmph, the RLF probability is almost unchanged. This is because with higher velocity the UE leaves the coverage of the serving cell faster, thus increasing the chance of an RLF. However, as the velocity keeps increasing beyond a threshold value of $50$ kmph, its effect on the RLF remains the same. Similarly, for the case of Rician fading, we observe an initial increase in RLF probability as the velocity increases up to $100$ kmph. However, as the velocity increases beyond $100$ kmph, we observe that the RLF probability starts decreasing. This is because the LoS provided by the Rician fading channel provides higher RSRP compared to that of Raleigh counterpart. As a result, the chance of condition \eqref{prep_offset} violation is less in Rician fading channel. Moreover, due to higher RSRP, the chance of RLF during $T_{prep}$ duration is also less, which increase the chances for successful handovers to new cells. This observation can be correlated with the result presented in Fig. \ref{hofv}, which shows variation of HOF probability with UE velocity. Here RLF has been considered as the primary cause of HOF, hence the HOF probability follows the same trend as RLF probability.

\begin{figure}[htbp]
    \centering
    \includegraphics[page=4,clip, trim=5cm 7.0cm 2.5cm 4cm,scale=.36]{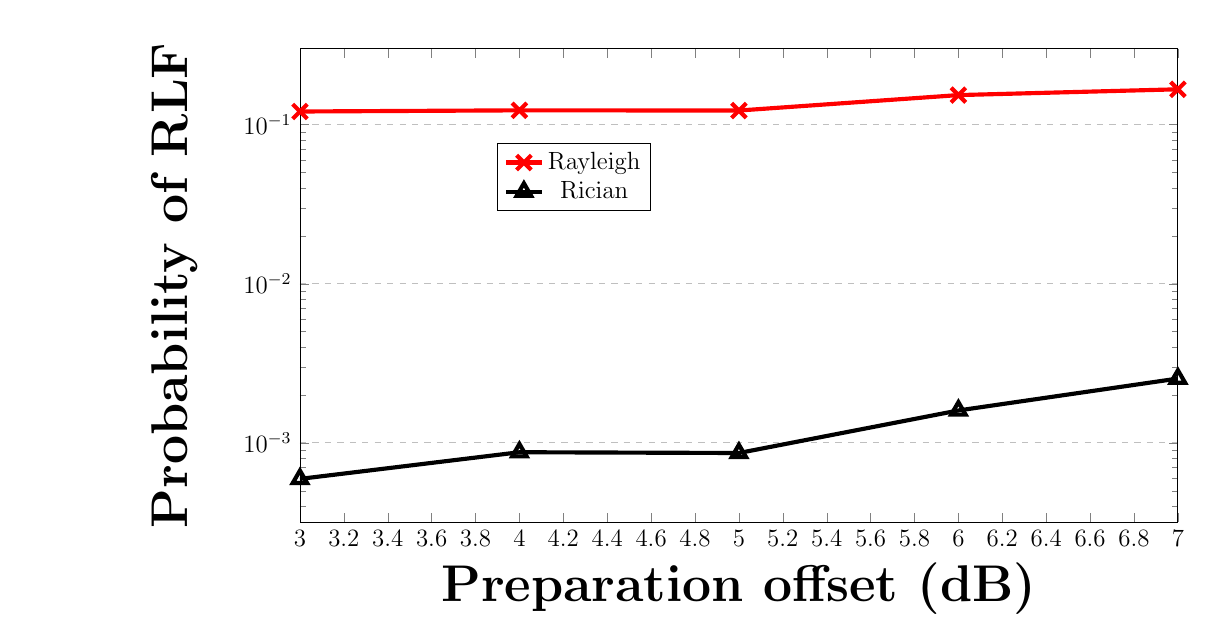}
    \caption{Probability of RLF (in log scale) vs. preparation offset [$\mathbf{T_{prep}=100\;ms}$, $\mathbf{T_{exec}=80\;ms}$, $\mathbf{O_{exec}=6}$ dB]}
    \label{rlfp}
\end{figure}

\begin{figure}[htbp]
    \centering
    \includegraphics[page=5,clip, trim=5cm 7.0cm 2.5cm 4cm,scale=.36]{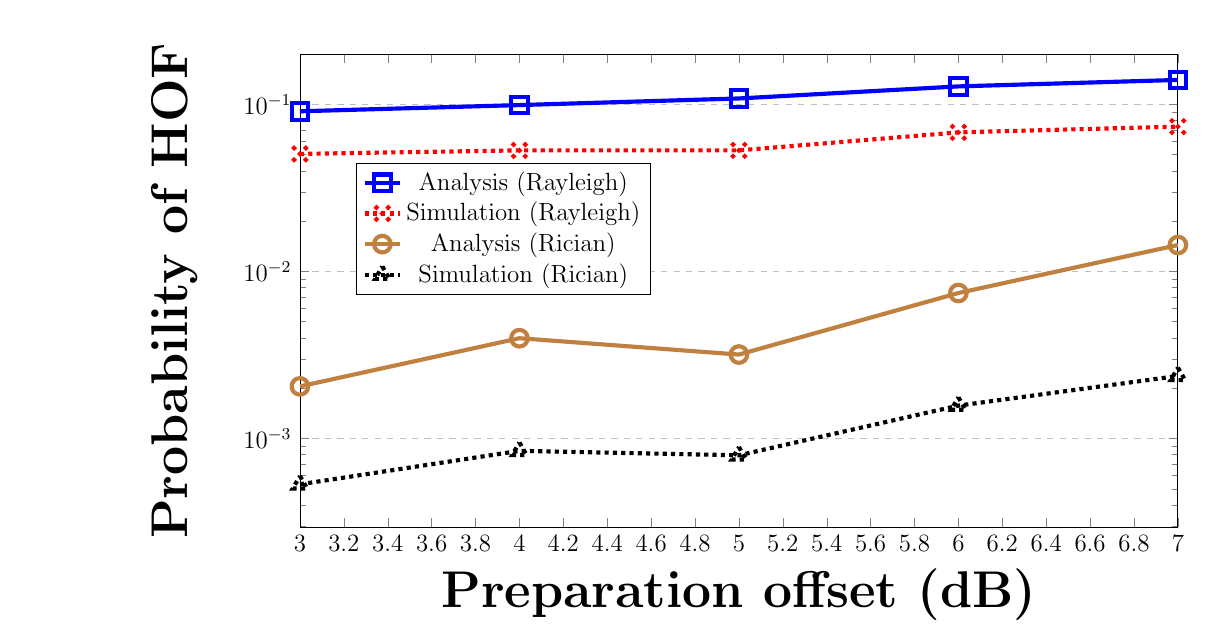}
    \caption{Probability of HOF (in log scale) vs. preparation offset [$\mathbf{T_{prep}=100\;ms}$, $\mathbf{T_{exec}=80\;ms}$, $\mathbf{O_{exec}=6}$ dB]}
    \label{hofp}
\end{figure}

Figs \ref{rlfp} and \ref{hofp}, depict the relationship of the probabilities of RLF and HOF with respect to $O_{prep}$. In both the figs, $O_{prep}$ varies from $3$ dB to $7$ dB with a step of $1$ dB. Irrespective of fading, we observe an increasing trend for both RLF and HOF. This is because increasing the value of $O_{prep}$ increases the chance of violating the handover preparation condition \eqref{prep_offset}. As a result, the UE  has to restart the handover preparation phase frequently, resulting in higher delay in initiating handover. Such delay results in higher RLFs and subsequent HOFs. It may be observed that the simulation result closely resembles the analytic result. Minor differences between simulation and analytic result is observed due to the approximations in the Markov model: number of states in handover preparation and handover execution phases have been computed as ceiling of $\frac{T_{prep}}{T_{sample}}$ and  $\frac{T_{exec}}{T_{sample}}$ which results in extra $2$ states. As a result, for both Raleigh and Rician fading cases, the HOF value obtained from analysis is lesser than the HOF value obtained through simulation. It may be observed that HOF probability for Raleigh fading cases are higher compared to that of Rician counterpart. This is because of less frequent violation of condition \eqref{prep_offset} due to higher RSRP in Rician fading channel. 

\begin{figure}[htbp]
    \centering
    \includegraphics[page=6,clip, trim=5cm 7.0cm 2.5cm 4cm,scale=.36]{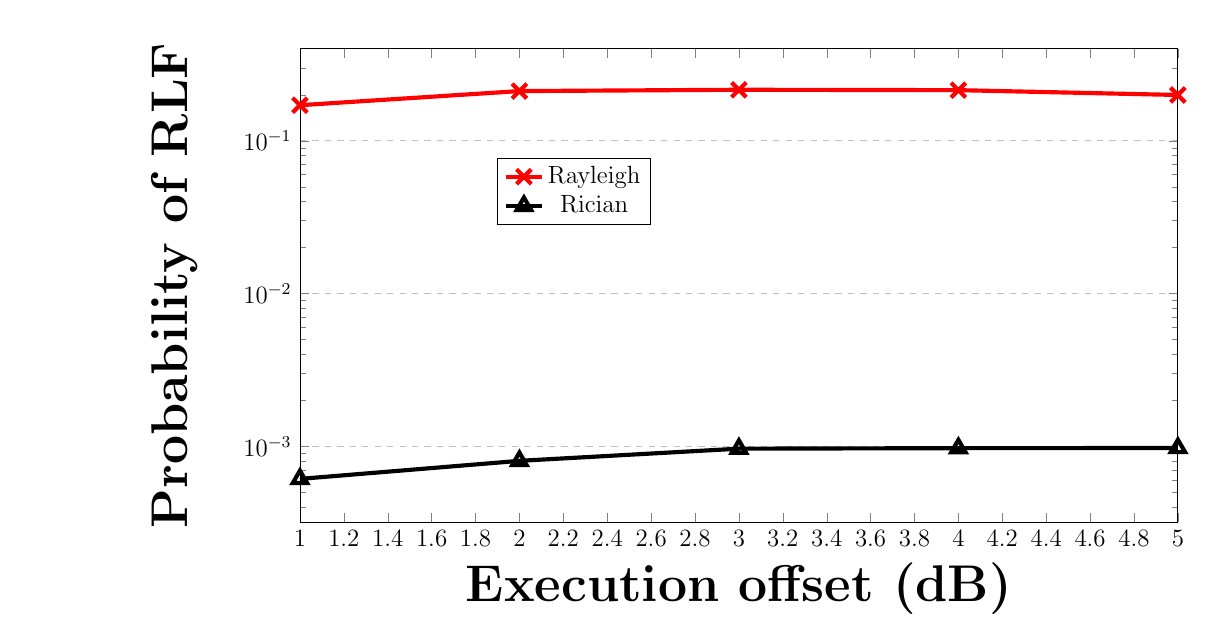}
    \caption{Probability of RLF (in log scale) vs. execution offset [$\mathbf{T_{prep}=100\;ms}$, $\mathbf{T_{exec}=80\;ms}$, $\mathbf{O_{prep}=1}$ dB]}
    \label{rlfe}
\end{figure}
\begin{figure}[htbp]
    \centering
    \includegraphics[page=7,clip, trim=5cm 7.0cm 2.5cm 4cm,scale=.36]{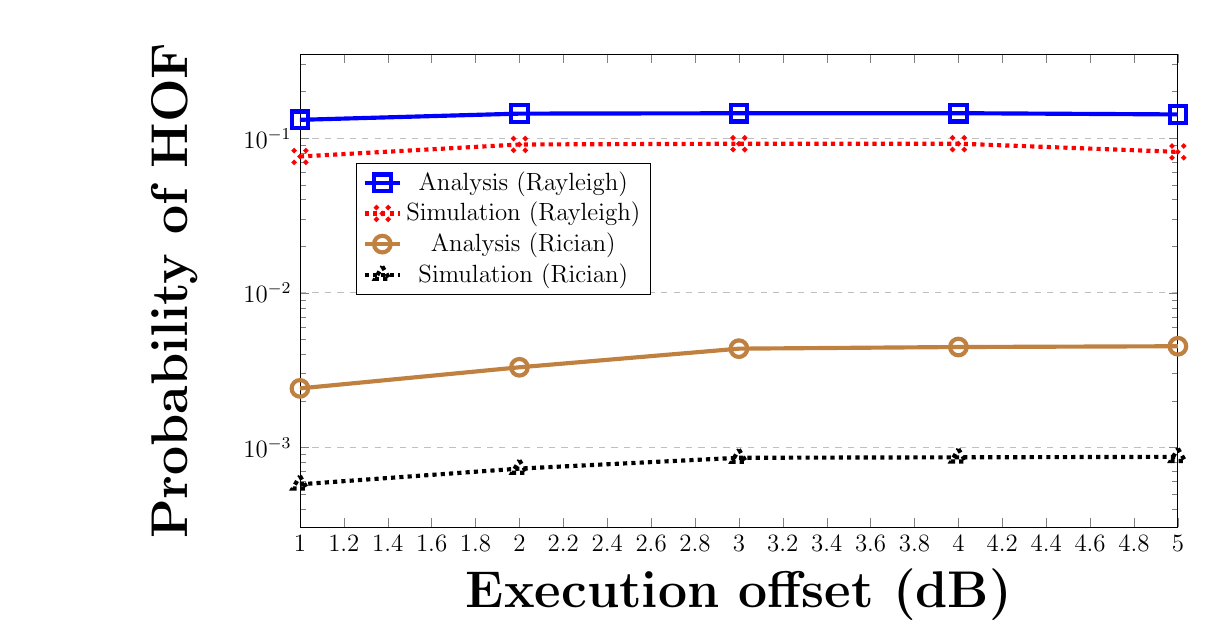}
    \caption{Probability of HOF (in log scale) vs. execution offset [$\mathbf{T_{prep}=100\;ms}$, $\mathbf{T_{exec}=80\;ms}$, $\mathbf{O_{prep}=1}$ dB]}
    \label{hofe}
\end{figure}

Figs \ref{rlfe} and \ref{hofe} depict the effect of $O_{exec}$ on RLF and HOF probabilities.  In both figs, the offset varies from $1$ dB to $5$ dB with a step of $1$ dB. For both Rayleigh and Rician fading channels, we observe an initial increase in the RLF and HOF probabilities. Higher offset values lead to a greater chance of the condition \eqref{exec_offset} to be violated. This results in longer waiting time for the condition \eqref{exec_offset} to be satisfied. The delay in handover caused by the increased waiting time results in RLF and subsequent HOF. It is interesting to observe that the trend stabilizes after $2$ dB because, by that time the handover execution phase has already begun and the UE has moved far enough from the serving cell such that RSRP from the target cell become higher. Consequently, the difference between the RSRP values of the serving and the target cell is large enough to hold condition \eqref{exec_offset}. Hence, no further decrease in RLF/HOF is observed as $O_{exec}$ increases beyond $2$ dB. The analytic result closely resembles the simulation result. This is because of the approximation made in analytic modelling as described before. The HOF probability for Raleigh fading is higher compared to that of Rician because, the condition \eqref{exec_offset} has higher chance of being violated for Raleigh fading channel compared to that of Rician counterpart. This because of the fact that the RSRP in Rician fading channel is more stable compared to that of Raleigh because of the presence of the LoS path as described above.

\paragraph*{Observation on $T_{exec}$ and $T_{prep}$} In this study, $T_{exec}$ and $T_{prep}$ has been set to $100$ ms and $80$ ms respectively. It has been observed from simulation that for higher $T_{exec}$ and $T_{prep}$ ($T_{exec} > 100$ ms and $T_{prep} >80$ ms), RLFs and subsequent HOFs are certain events because of higher handover delay. For lower values of $T_{exec}$ and $T_{prep}$, effect of fading is not visible because number of sampling instances during $T_{exec}$ and $T_{prep}$ decreases. 

\subsection{Analysis of handover latency and handover packet loss}
In this section, we evaluate the effect of $O_{prep}$, $O_{exec}$, $T_{prep}$ and $T_{exec}$ on handover latency and handover packet loss. Handover latency is defined as the expected time between beginning of handover preparation phase and the end of handover execution phase, i.e., $\mathbf{E}[t_6 + T_{exec}-t_1]$ (in Fig. \ref{cho_mech}). Accordingly, handover packet loss is defined as the number of packets communicated in down-link direction during handover, and therefore lost. Handover packet loss has been computed as:
\begin{equation}
    \text{Handover packet loss }= \lambda \times \mathbf{E}[t_6 + T_{exec}-t_1].
\end{equation}
\noindent where $\lambda$ is the packet arrival rate and $\mathbf{E}[t_6 + T_{exec}-t_1]$ is the expected value of the handover latency which has been measured from simulation traces.

\begin{figure}[htbp]
    \centering
    \includegraphics[page=18,clip, trim=5cm 6.0cm 2.5cm 4cm,scale=.36]{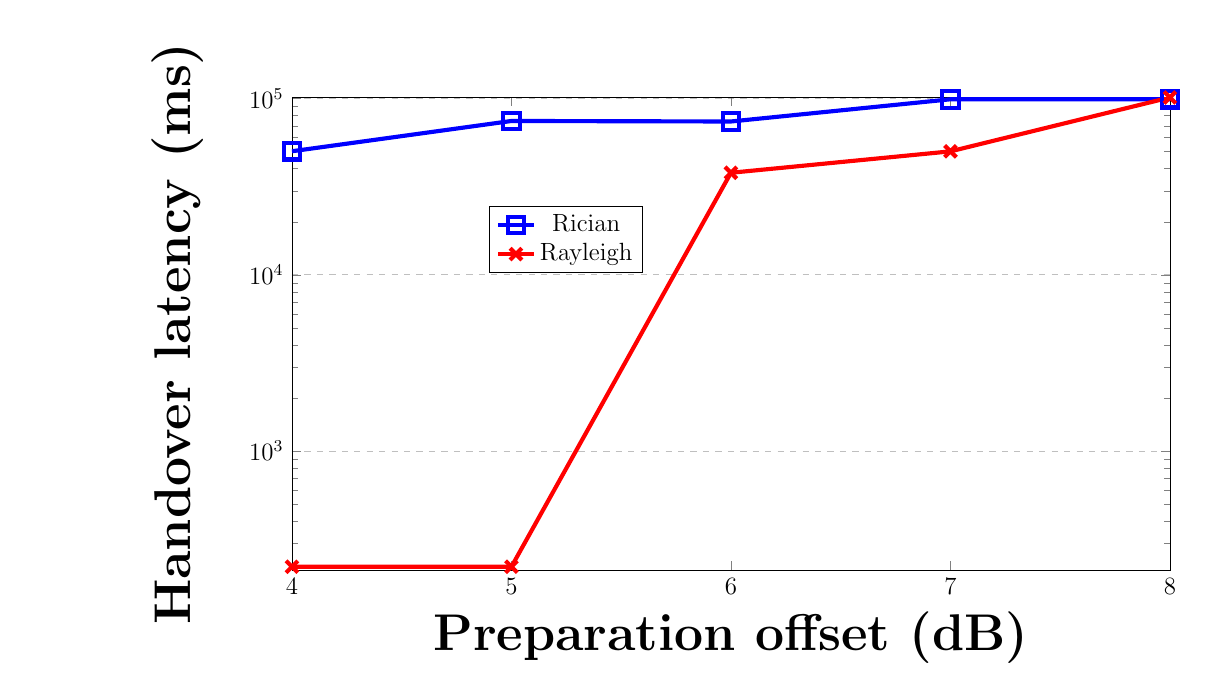}
    \caption{Handover latency (in log scale) vs. preparation offset [$\mathbf{T_{prep}=100\;ms}$, $\mathbf{T_{exec}=80\;ms}$, $\mathbf{O_{exec}=6}$ dB]}
    \label{lat_prep}
\end{figure}

\begin{figure}[htbp]
    \centering
    \includegraphics[page=21,clip, trim=5cm 7.0cm 2.5cm 4cm,scale=.36]{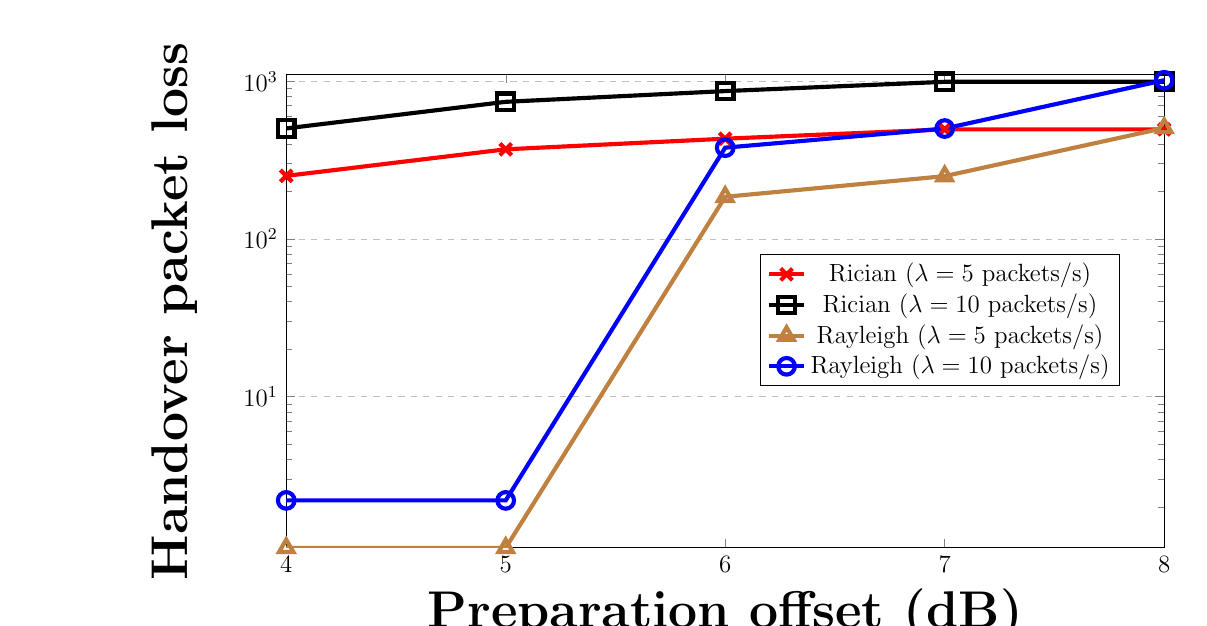}
    \caption{Handover packet loss (log scale) vs. preparation offset [$\mathbf{T_{prep}=100\;ms}$, $\mathbf{T_{exec}=80\;ms}$, $\mathbf{O_{exec}=6}$ dB]}
    \label{pkl_prep}
\end{figure}

Fig. \ref{lat_prep} depicts the effect of $O_{prep}$ on handover latency and handover packet loss. Here the offset varies from $4$ dB to $8$ dB with a step of $1$ dB. For both Rayleigh and Rician fading scenarios, we observe an increasing trend in the handover latency with increasing values of $O_{prep}$. This is because for higher values of $O_{prep}$, the RSRP of the target cell also needs to be higher for condition \eqref{prep_offset} to be satisfied. Fluctuations in the RSRP of both the target and the source cells due to fading often result in violation of condition \eqref{prep_offset}, resulting in restarting the handover preparation phase. The chances of violating condition \eqref{prep_offset} for a given magnitude of fluctuation in RSRP value of the target gNB increases with increasing value of $O_{prep}$. The increasing number of violations in condition \eqref{prep_offset} results in an increased number of restarts of the handover preparation phase, and thereby increasing the handover latency. No specific trend is observed for Rician fading scenario. However, it may also be observed that the handover latency for the Rician fading is higher compared to that for Rayleigh fading. The presence of an LoS component for the case of Rician fading improves the coverage of the serving cell which results in a stronger RSRP at the UE as compared to the Rayleigh fading case. Therefore, in case of Rician fading, after an initial violation of condition \eqref{prep_offset}, the UE needs to move farther from the source cell to satisfy condition \eqref{prep_offset} again as compared to Rayleigh counterpart. As a result, handover delay in case of Rician channel is higher compared to that of Rayleigh channel. Fig. \ref{pkl_prep} shows the variation of packet loss rate with $O_{prep}$. Packet loss shows an increasing trend with increasing $O_{prep}$. This is because of the increasing trend of the handover latency with increasing $O_{prep}$ as described above. Moreover, the handover packet loss is higher for higher values of $\lambda$. This is because of the higher number of packet arrival during handover duration.

\begin{figure}[htbp]
    \centering
    \includegraphics[page=19,clip, trim=5cm 6.0cm 2.5cm 4cm,scale=.36]{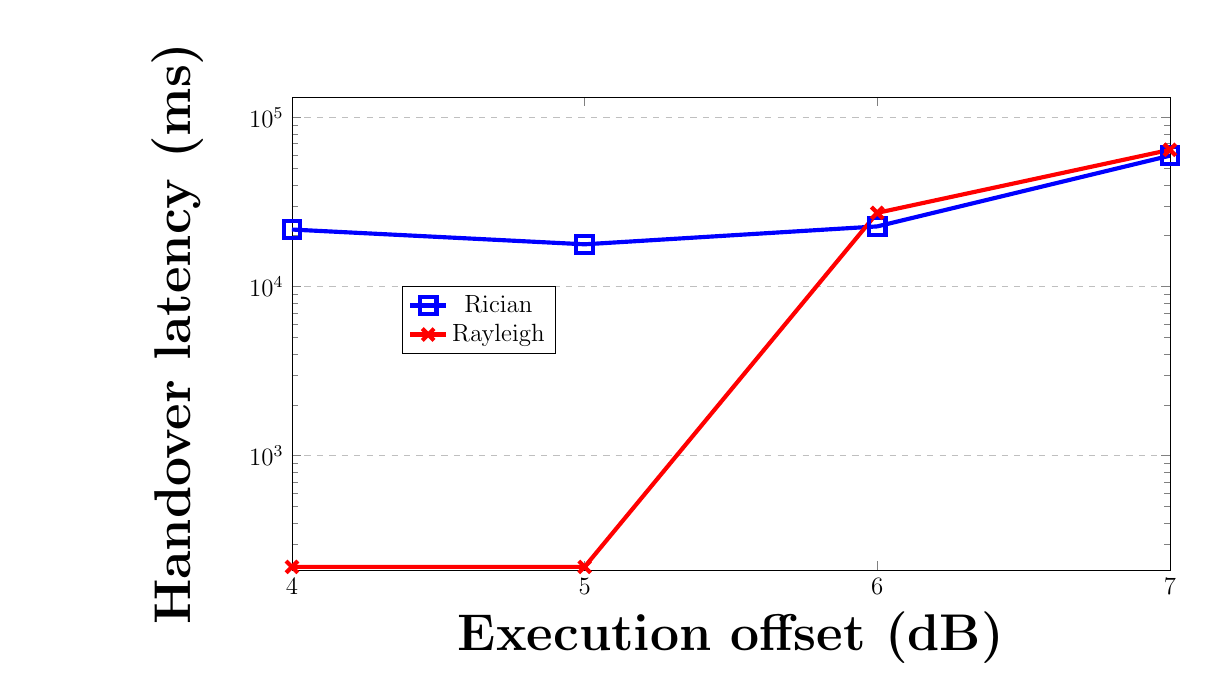}
    \caption{Hanover latency (log scale) vs. execution offset [$\mathbf{T_{prep}=100\;ms}$, $\mathbf{T_{exec}=80\;ms}$, $\mathbf{O_{prep}=1}$ dB]}
    \label{lat_exec}
\end{figure}

\begin{figure}[htbp]
    \centering
    \includegraphics[page=22,clip, trim=5cm 7.0cm 2.5cm 4cm,scale=.36]{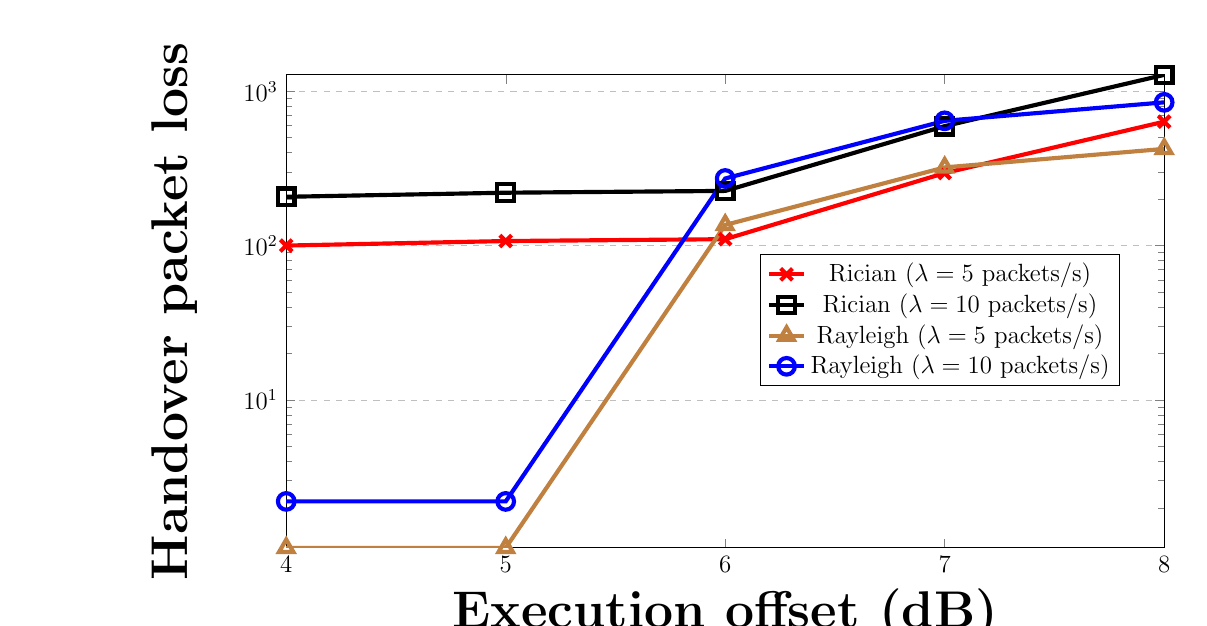}
    \caption{Hanover packet loss (log scale) vs. execution offset [$\mathbf{T_{prep}=100\;ms}$, $\mathbf{T_{exec}=80\;ms}$, $\mathbf{O_{prep}=1}$ dB]}
    \label{pkl_exec}
\end{figure}

Figs \ref{lat_exec} and \ref{pkl_exec} depict the effect of $O_{exec}$ on handover latency. In both the figs, the offset varies from $4$ dB to $7$ dB with a step of $1$ dB. It is to be noted that similar to the handover preparation phase, violation of condition \eqref{exec_offset} results in a restart of the handover execution phase. For both Rayleigh and Rician fading scenario, the handover latency monotonically increases with increasing values of $O_{exec}$. This is because the chances of violating condition \eqref{exec_offset} increases with increasing $O_{exec}$, and thereby increasing the number of restarts of the handover execution phase leading to an increased handover latency. It may also be observed that the handover latency for the Rician fading is higher compared to that of Rayleigh fading scenario. Similar to Fig. \ref{lat_prep}, after an initial violation of condition \eqref{exec_offset}, the UE needs to move a higher distance from the source cell for condition \eqref{exec_offset} to be satisfied in case of the Rician channel as compared to Rayleigh counterpart. As a result, handover latency in Rician fading channel is higher compared to that of Raleigh fading scenario. Fig. \ref{pkl_exec} shows that handover packet loss monotonically increases with increasing $O_{exec}$ and $\lambda$. The reason behind is similar to that described for $O_{prep}$ (Figs \ref{lat_prep} and \ref{pkl_prep}).

\section{Conclusions} \label{con}

In this work, an effort has been initiated to explore the effect of various mobility parameters such as $O_{exec}$, $O_{prep}$, $T_{exec}$ and $T_{prep}$ under different UE velocity and channel fading characteristics. We have considered handover failure probability, handover latency and handover packet loss as performance evaluation metrics. A Markov model based analytic framework has been proposed to characterize handover failure for conditional handover mechanism. Results obtained from the analytic model has been validated against extensive simulation results. Our study reveal that optimal configuration of $O_{exec}$, $O_{prep}$, $T_{exec}$ and $T_{prep}$ is actually conditional on underlying UE velocity and channel fading characteristics. This study will be helpful for the mobile operators to choose appropriate thresholds of the mobility parameters based on channel condition and UE velocities.

\section{Acknowledgments}
This research is supported by Opportunity for Undergraduate Research (OUR) scheme (Project no: OUR20230014) of Shiv Nadar Institution of Eminence (SNIoE), Delhi, NCR. A part of this work has been submitted to SNIoE as an internal project report. 

\bibliographystyle{IEEEtran}
\bibliography{sample-base}
\end{document}